\documentclass[conference]{IEEEtran}
\usepackage{bm}
\usepackage{amssymb}

\makeatletter
\def\ps@headings{%
\def\@oddhead{\mbox{}\scriptsize\rightmark \hfil \thepage}%
\def\@evenhead{\scriptsize\thepage \hfil \leftmark\mbox{}}%
\def\@oddfoot{}%
\def\@evenfoot{}}
\makeatother
 \makeatletter
 \def\ps@headings{%
 \def\@oddhead{\mbox{}\scriptsize\rightmark \hfil \thepage}%
 \def\@evenhead{\scriptsize\thepage \hfil \leftmark\mbox{}}%
 \def\@oddfoot{}%
 \def\@evenfoot{}}
 \makeatother
 \usepackage{array}
 \usepackage[svgnames]{xcolor} 
 \usepackage{mdwmath}
 \usepackage{subfigure}
  \usepackage{bm}
 \usepackage{tipa}
  \usepackage{hyperref}
  \usepackage{framed}
  \usepackage{float}
 \usepackage{algorithm}
\usepackage{algorithmic}
\usepackage{overpic}
\usepackage{framed}
\usepackage{url}
\usepackage{graphicx}
\usepackage{times}
\usepackage{amsmath}
\PassOptionsToPackage{hyphens}{url}\usepackage{hyperref}
 \newtheorem{theorem}{Theorem}[section]
\newtheorem{lemma}[theorem]{Lemma}

\begin{document}
\title{Distributed Load Management in Anycast-based CDNs}

%
%
%
%
%
%

\author{
\IEEEauthorblockN{Abhishek Sinha \IEEEauthorrefmark{1}, Pradeepkumar Mani\IEEEauthorrefmark{2}, 
Jie Liu\IEEEauthorrefmark{2}, Ashley Flavel\IEEEauthorrefmark{2} and David A. Maltz\IEEEauthorrefmark{2}}
\IEEEauthorblockA{\IEEEauthorrefmark{1}MIT, \IEEEauthorrefmark{2}Microsoft\\
Email: sinhaa@mit.edu, 
\{prmani, Jie.Liu, ashleyfl, dmaltz\}@microsoft.com}

}
\maketitle

\begin{abstract}
Anycast is an internet addressing protocol where multiple hosts share the same IP-address. A popular architecture for modern Content Distribution Networks (CDNs) for geo-replicated HTTP-services consists of multiple layers of proxy nodes for service and co-located  DNS-servers for load-balancing on different proxies. Both the proxies and the DNS-servers use anycast addressing, which offers simplicity of design and high availability of service at the cost of partial loss of routing control. Due to the very nature of anycast, load-management decisions by a co-located DNS-server also affects loads at nearby proxies in the network. This makes the problem of distributed load management highly challenging. In this paper, we propose an analytical framework to formulate and solve the load-management problem in this context. We consider two distinct algorithms. In the first half of the paper, we pose the load-management problem as a convex optimization problem. Following a dual decomposition technique, we propose a fully-distributed load-management algorithm by introducing \emph{FastControl} packets. This algorithm utilizes the underlying anycast mechanism itself to enable effective coordination among the nodes, thus obviating the need for any external control channel. In the second half of the paper, we consider an alternative greedy load-management heuristic, currently in production in a major commercial CDN. We study its dynamical characteristics and analytically identify its operational and stability properties. Finally, we critically evaluate both the algorithms and explore their optimality-vs-complexity trade-off using trace-driven simulations.
\end{abstract}

\section{Introduction}

Content Distribution Networks (CDN) are \emph{de facto} architectures to transparently reduce latency between the end-users and geo-replicated Internet services. Edge-servers with cached contents serve as proxies to intercept some user requests and return contents without a round-trip to the data centers. Routing user-requests to the optimal proxies remains a challenge in managing modern CDN. Routing to a remote proxy may introduce extra round-trip delay, whereas routing to an overloaded proxy may cause the request to be dropped. 

Anycast is a relatively new paradigm for CDN-management and there are  already several commercial CDN in place today using anycast \cite{NSDI_paper}, \cite{engel1998using}, \cite{swildens2009global}. With anycast, multiple proxy-servers share the same IP-address. Anycast relies on routing protocols (such as, BGP) to route service-requests to any one of the geo-replicated proxies, over a \emph{cheap} network-path \cite{zaumen2000load}. Anycast based mechanisms have the advantage of being simple to deploy and maintain.  Being available as a service in IPv6 networks, no global topology or state information are required for its use \cite{basturk1997using}. 

Although anycast routing can simplify the system-design and provide a high level of availability to users \cite{sarat2006use}, it comes at the cost of partial loss of routing-control. This is because, a request may be routed to any one of the different geo-replicated proxies, determined solely by the routing protocol and network state. Since this routing is not under the control of the CDN-operator, the request may end up in an already overloaded proxy, deteriorating the situation further. There have been several attempts in the literature to tackle this lack of load-awareness issue with network-layer anycast. Papers \cite{miura2001server}, \cite{hashim2005active} consider ``Active Anycast" where additional intelligence is incorporated into the routers based on RTT and network congestion. It is specifically targeted to reduce pure latency rather than server overload,  thus yielding sub-optimal performance in a CDN setting. Alzoubi \emph{et al.} \cite{Alzoubi} poses the anycast load-management problem as a General Assignment Problem, which is NP-hard to solve in general. The paper \cite{jaseemuddin2006te} proposes a new CDN architecture which balances server-load and network-latency via detailed traffic engineering. 

In this paper, we focus our attention to the state-of-the-art CDN-architecture such as \emph{FastRoute} \cite{NSDI_paper}, which uses DNS-based-redirection by the co-located DNS-servers for overload control in the local proxies. Since DNS is the first point-of-contact of users to the internet, DNS-redirection is a popular and effective way to mitigate overload \cite{pang2004responsiveness}, \cite{shaikh2001effectiveness}. In this architecture, the proxies are arranged in layers of anycast rings and DNS is responsible for moving load across different layers. In the sequel, each proxy and the co-located DNS unit will be referred to simply as a \emph{node}. See Figures \ref{fig:layers} and \ref{dual_fig} for an overview of the architecture.
 
An \emph{overload} is said to occur when any individual proxy receives more requests than it can process. Since DNS is the primary control knob in this architecture, it is responsible for redirecting traffic to the next layer to alleviate overload. A fundamental problem with this approach is that not all users, that hit a given proxy, can be redirected by the co-located DNS. This is because an user's ISP could be obtaining a DNS-response from a DNS-server different than the  co-located proxy. Hence, intuitively, the ability for a DNS-server to control overload at the corresponding co-located proxy depends on the fraction of oncoming traffic to the proxy that are routed by the co-located DNS-server. Informally, we refer to the above quantity as the \emph{self-correlation} \cite{NSDI_paper} of a given node. A formal definition of correlation and associated quantities will be given in section \ref{sec:model}.

Poor self-correlation could impair a node's ability to control overload in isolation. Hence, successful load management in layered CDN should involve coordinated action by DNS-servers in multiple nodes to alleviate overload. Thus the problem reduces to the DNS-plane determining the appropriate offload or redirection probabilities at each node to move traffic from the overloaded proxies to the next layer. This control-decision could be based on variety of information such as load on each proxy, DNS-HTTP correlation etc. From a practical point of view, not all of these quantities are easily measured and communicated to wherever is needed. Thus a centralized solution is not practically feasible and the challenge is to design a provably optimal, yet completely distributed load management algorithm.

Our key contributions in this paper are as follows:
\begin{itemize}
\item In section \ref{sec:model}, we present a simplified mathematical model for anycast-based load-management in modern CDNs. Our model is general enough to address the essential operational problems faced by the CDN-operators yet tractable enough to draw meaningful analytical conclusions about their operational characteristics.

\item  In section \ref{optimization}, we pose the load management problem as a convex optimization problem and derive a \emph{dual} algorithm to solve it in a distributed fashion. The key to our distributed implementation is the  \emph{Lagrangian decomposition} and the use of \emph{FastControl} packets, which exploits the underlying anycast capability to enforce coordination among the nodes in a distributed fashion. To the best of our knowledge, this is the first instance of such a decomposition technique employed in the context of load-management in CDN.   

\item In section \ref{heuristic}, we consider an existing heuristic currently in operation in \emph{Microsoft Azure}, a major commercial CDN \cite{azure}. We model the dynamics of the heuristic using non-linear system-theory and derive its several important operational characteristics. To provide further insight, a two-node system is analyzed in detail and it is shown, rather surprisingly, that given the ``self-correlations" of the nodes are sufficiently high, this heuristic algorithm is able to control an incoming-load of \emph{any} magnitude, however large. Unfortunately, this theoretical guarantee breaks down once this correlation property does not hold. In this case, the dual algorithm developed in section \ref{optimization} performs arbitrarily better than the heuristic. 

	\item In section \ref{simulations}, we critically evaluate relative performances of both the optimal and heuristic algorithms through extensive simulations. Our simulation is trace-driven in the sense we use real correlation parameters  collected over months from an operational CDN \cite{NSDI_paper}. 
\end {itemize}

\section{System Model} \label{sec:model}

\begin{figure}[]
\center
\begin{overpic}[width=0.5\textwidth]{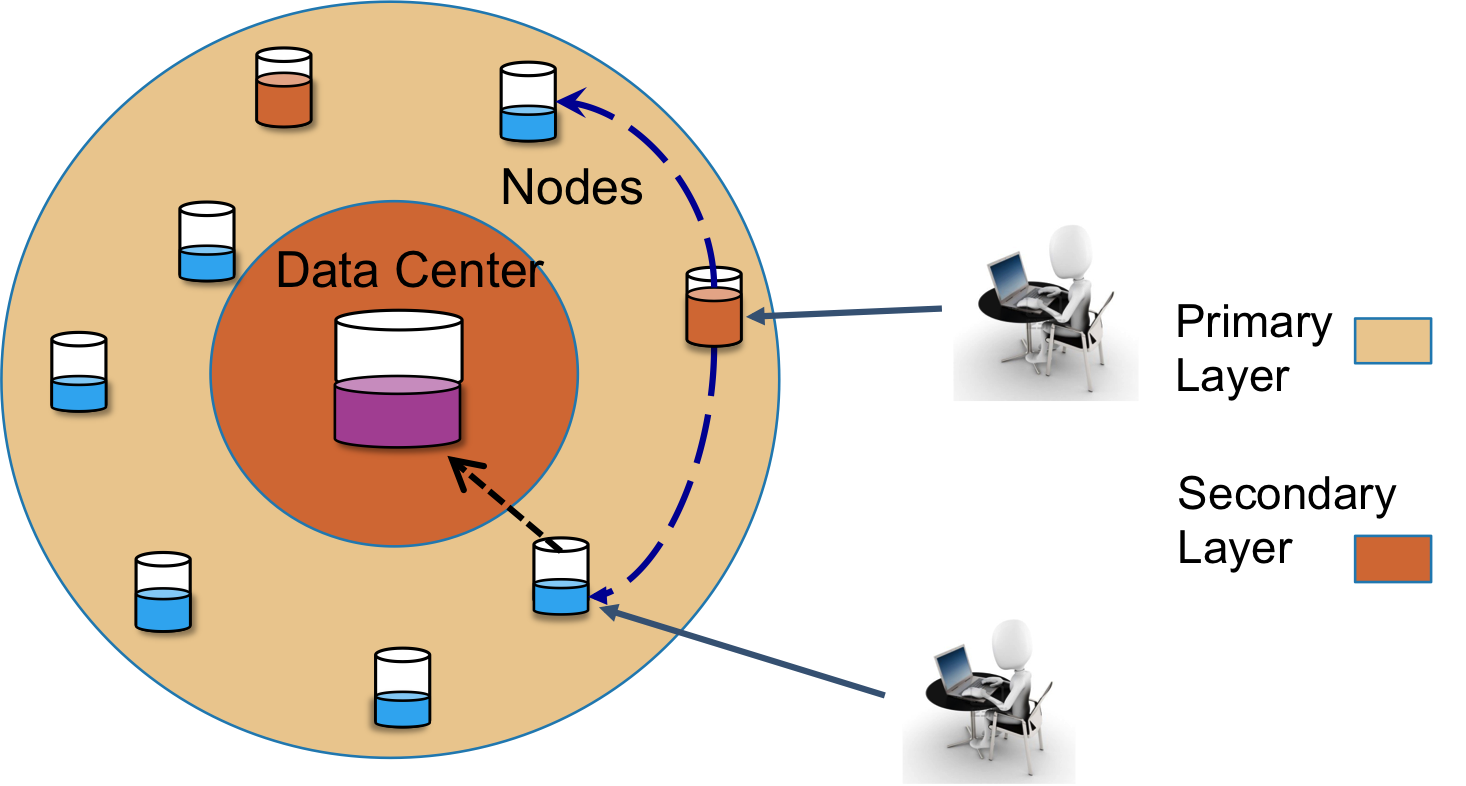}
\put(20,9){{$L_1$ }}
\put(25,18){{$L_2$}}
\put(75,5){\small{Users}}
\end{overpic}
\caption{A toy CDN with two anycast layers.  Solid arrows denote user connections, while dotted arrows denote the effect of diverting traffic by overloaded nodes.} 
\label{fig:layers}
\end{figure}
\indent \emph{Nodes and Layers:} We consider a CDN system consisting of two layers: \emph{primary} (also referred to as $L_1$) and \emph{secondary} (also referred to as $L_2$), as shown in Fig \ref{fig:layers}. The primary layer hosts a total of $N$ nodes. Each node consists of a collocated DNS and a proxy server. See Figure \ref{dual_fig} for a schematic diagram. The proxy servers are the end-points for HTTP-requests. The $i$\textsuperscript{th} proxy has a (finite) processing-capacity of $T_i$ requests per unit time. The secondary layer ($L_2$) consists of a large data-center, with practically infinite processing-capacity. Since the proxies are distributed throughout the world, an average user is typically located near to some proxy, resulting in relatively small round-trip latency. On the other hand, the data-center in $L_2$ is typically located further from the average user-base, resulting in significantly high round-trip latency. Total response time for a user is the sum of round-trip latency and processing time of the server (either proxy or the data-center), which we would like to minimize. \\
\indent \emph{Anycast Addressing :} \emph{All} proxies in the primary layer share the same IP-address $\mathcal{I}_1$ and the data-center in $L_2$ has a distinct IP-address $\mathcal{I}_2$. This method of multiple proxies sharing the same IP-addresses is known as \emph{anycast} addressing and is widely used for geo-replicated cloud services \cite{cloudflareanycast}, \cite{Alzoubi}.\\
\indent \emph{Control decisions :} When a DNS-query arrives at a node $i$, asking for an IP-address of an HTTP-server, the collocated DNS-server decides whether to return the address $\mathcal{I}_2$ (meaning, to offload the incoming request to the data-center) or the anycast address $\mathcal{I}_1$ (meaning, not to offload the incoming request to the far-away data-center and serve it in some proxy in the primary layer itself). This (potentially randomized) binary decision could be based on the following observables available to the node $i$ : 
\begin{enumerate}
 \item \emph{DNS-query} arrival rate at node $i$'s co-located DNS-server, given by $A_i$ requests per unit time\footnote{We assume that $A_i$'s are piecewise constant and do not change during the transient period of the load-management algorithms discussed here.}. Each DNS-query accounts for a certain amount of user HTTP-load which is influenced by the DNS-response at node $i$; We normalize $A_i$ so that each unit of $A_i$ corresponds to a unit of HTTP-load; As an example, if 10\% of DNS responses at node $i$ returns the address $\mathcal{I}_2$, then the total load shifted to $L_2$ due to node $i$'s DNS is $0.1A_i$.
 \item \emph{Content-request} arrival rate (or simply, the \emph{load}) at the $i$\textsuperscript{th} collocated proxy, given by $S_i(t)$ requests per unit time. This is clearly a result of $A_i$'s and different nodes' offload-decisions (Eqn.\eqref{S_eqn}). Since we focus on web-applications, we will use the phrase content-requests and HTTP-requests interchangeably. The observable $S_i(t)$ is available locally at node $i$. We assume that the workloads introduced by different requests are roughly the same.

\end{enumerate}

\indent \emph{Anycasting and inter-node coupling:} When a DNS-query arrives at a node $i$ and the corresponding DNS-server returns the $L_1$ anycast-address $\mathcal{I}_1$, the request may be routed to any of one of the $N$ proxy-nodes in the primary layer for service, depending on the corresponding ISP's routing policy, network congestion and other random factors. We assume that a typical DNS-query, routed to $L_1$ by the DNS at node $i$, arrives for service at node $j$'s proxy with probability $C_{ij}> 0$, where
\begin{eqnarray} \label{corr_mat_prop}
 \sum_{j=1}^{N}C_{ij}=1, \hspace{10pt}\forall i=1,2,\ldots, N
\end{eqnarray}
 The matrix $\bm{C}\equiv [C_{ij}]$ can be determined empirically by setting up an experiment similar to the one described in \cite{shaikh2001effectiveness}. In this paper, we primarily focus on effects that arise from non-trivial couplings among the nodes. 

\indent \emph{Input-Output Equations:} Assume that, due to action of some control-strategy $\pi$, the collocated DNS at node $i$ decides to divert $1-x_i^\pi(t)$ fraction of incoming DNS-queries (given by $A_i$) to Layer $L_2$ at time $t$ ($0\leq x_i^\pi(t) \leq 1$). Thus it routes $x_i^\pi(t)$ fraction of the incoming requests to different proxies in the layer $L_1$. Hence the total HTTP-request arrival rate, $S_i(t)$, at the $i$\textsuperscript{th} proxy may be written as 
\begin{eqnarray} \label{S_eqn}
 S_i(t)=\sum_{j=1}^{N} C_{ji}x_j^\pi(t)A_j, \hspace{10pt}\forall i=1,2,\ldots, N
\end{eqnarray}
A local control strategy $\pi$ is identified by a collection of mappings $\bm{\pi}=\bigg(x^\pi_i(\cdot),i=1,2,\ldots,N\bigg)$, given by $\bm{x}^\pi_i: \Omega^t_i \times t \to [0,1]$, where $\Omega^t_i$ is the set of all observables at node $i$ up to time $t$.
\section{An Optimization Framework} \label{optimization}
\subsection{Motivation} 
In our context, the central objective of a load-management policy $\pi$ is to route as few requests as possible to the secondary layer (due to its high round-trip latency), without overloading the primary-layer proxies (due to their limited capacities). Clearly, these two objectives are at conflict with each other and we need to find a suitable compromise. The added difficulty, which makes the problem fundamentally challenging is that, the nodes are autonomous agents and take their redirection decisions on their own, based on their local observables only. As an example, a simple locally-greedy heuristic for node $i$ could be to redirect requests to $L_2$ (i.e. decrease $x_i(t)$) whenever its co-located proxy is overloaded (i.e., $S_i(t)>T_i$) and redirect requests to $L_1$ (i.e., increase $x_i(t)$) whenever the proxy is under-loaded (i.e., $S_i(t)<T_i$). This forms the basis of the strategy adopted in \cite{NSDI_paper}. \\
 This greedy strategy appears to be quite appealing for deployment, due to its extreme simplicity. However, in the next subsection \ref{u_o}, we show by a simple example that, in the presence of significantly high cross-correlations among the nodes, this simple heuristic could lead to an  \emph{uncontrollable overload situation}, an extremely inefficient operating point with degraded service quality. This example will serve as a motivation to come up with a more efficient distributed load-management algorithm, that we develop subsequently. \\

\subsection{Locally Uncontrollable Overload: An Example} \label{u_o}
 Consider a CDN, hosting only two nodes $a$ and $b$ in the primary layer, as shown in Figure \ref{two-node_fig}. The (normalized) DNS-query arrival rates to the nodes $a$ and $b$ are  $A_a=A_b=1$. Suppose the processing capacities (also referred to as \emph{thresholds}) of the corresponding proxies are $T_a=T_b=0.7$. With the correlation values shown in figure \ref{two-node_fig}, the HTTP-request arrival rates (load) to the proxies at $a$ and $b$ are given as
 \begin{eqnarray}
  S_a(t)=0.1x_a(t)+0.5x_b(t)\\
  S_b(t)=0.9x_a(t)+0.5x_b(t)
 \end{eqnarray}
Since $0\leq x_a(t), x_b(t)\leq 1$, it is clear that 
\begin{eqnarray*}
	S_a(t) \leq 0.1\times 1 + 0.5 \times 1=0.6<0.7=T_a, \hspace{5pt} \forall t
\end{eqnarray*}
 Thus, the proxy at node $a$ will be under-loaded irrespective of the load-management policy $\pi$ in use. Consequently, under the greedy-heuristic (formally, algorithm \ref{algo_greedy} in Section \ref{heuristic}), the collocated DNS-server at node $a$ will \emph{greedily} increase its $L_1$ redirection probability $x_a(t)$ such that $x_a(t) \nearrow 1$ in the steady-state (note that, node $a$ acts independently as it does not have node $b$'s loading information). This, in turn, overloads the proxy in node $b$ because the steady-state HTTP-load at proxy $b$ becomes 
\begin{eqnarray*}
	S_b(\infty) &=& 0.9 x_a(\infty)+0.5x_b(\infty)\\
	&=& 0.9\times 1 + 0.5x_b(\infty)\\
	&\geq& 0.9 >0.7=T_b.
\end{eqnarray*}
\begin{figure}
\small
 \centering 
  \begin{overpic}[scale=0.7]{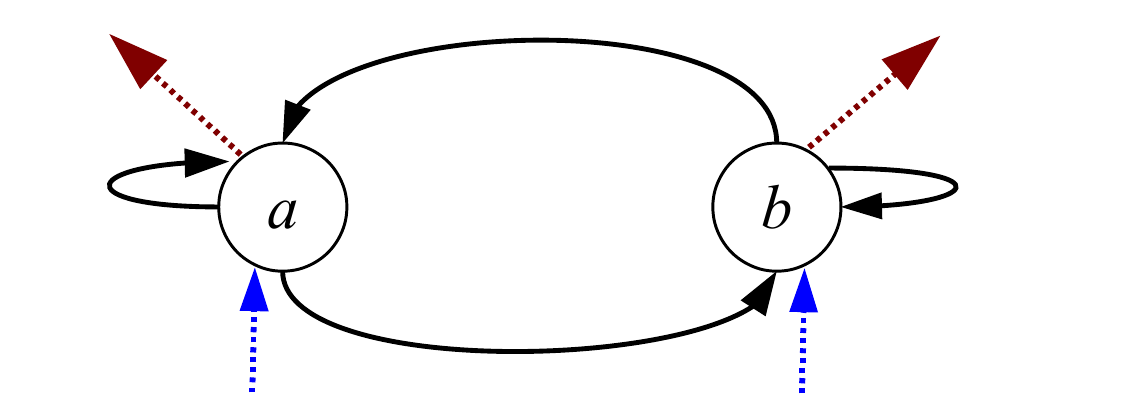}
  \put(10,13){$C_{aa}$}
  \put(8,9){$=0.1$}
  \put(38,-1){$C_{ab}=0.9$}
  \put(38,35){$C_{ba}=0.5$}
  \put(74,12){$C_{bb}=0.5$}
  \put(18,-2){$A_a=1$}
  \put(67,-2){$A_b=1$}
  \put(30,21){$T_a=0.7$}
  \put(51,21){$T_b=0.7$}
  \put(45,10){$L_1$}
  \put(4,43){Traffic to}
  \put(9,39){$L_2$}
  \put(5,35){$A_a(1-x_a(t))$}
  \put(71,43){Traffic to}
  \put(76,39){$L_2$}
  \put(72,35){$A_b(1-x_b(t))$}
   \end{overpic}
\caption{A two-node system illustrating locally uncontrollable overload at the node $b$}
\label{two-node_fig}
\end{figure}
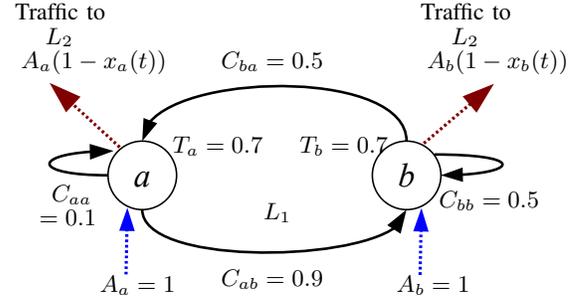
 Since node $b$ is overloaded in the steady-state, under the action of the above greedy heuristic, it will (unsuccessfully) try to avoid the overload by offloading the incoming DNS-queries to $L_2$, as much as it can, by letting $x_b(t)\searrow 0$. Thus the steady-state operating point of the algorithm will be $x_a(\infty)=1,x_b(\infty)=0$, with node $b$ overloaded. It is interesting to note that poor self-correlation of node $a$ ($C_{aa}=0.1$) causes the other node $b$ to overload, even under the symmetric DNS-query arrival patterns. Also, this conclusion does not depend on the detailed control-dynamics of the offload probabilities (viz. the instantaneous values if $\dot{x}_1(t)$ and $\dot{x}_2(t)$). Since the overload condition at the node-$b$ can not be overcome by isolated action of node-$b$ itself, we say that node-$b$ is undergoing a \emph{locally uncontrollable overload situation}. \\\\
 From the point-of-view of the overall system, this is an extremely inefficient operating condition, because a large fraction of the incoming requests either gets dropped or delayed due to the overloaded node-$b$. This poor operating point could have been avoided provided the nodes somehow mutually co-ordinate their actions\footnote{Another trivial solution to avoid overload could be to offload all traffic from all nodes to $L_2$, i.e. $x_i(t)=0, \forall i, \forall t$. However, this is highly inefficient because it is tantamount to not using the primary servers at all.}. 
It is not difficult to realize that the principal reasons behind the locally uncontrollable overload situation in the above example are as follows:
 \begin{itemize}
  \item (1) distributed control with local information 
  \item (2) poor self-correlation of node $a$ ($C_{aa}=0.1$)
 \end{itemize}
The factor (1) is fundamentally related to the distributed nature of the system and requires coordinations among the nodes. In our distributed algorithm [\ref{algo_dual}] we address this issue by introducing the novel idea of \emph{FastControl} packets. This strategy does not require any explicit state or control-information exchange.\\
Regarding the factor (2), we intuitively expect that the local greedy-heuristic should work well if the self-correlation of the nodes (i.e. $C_{ii}$) are not too small. In this favorable case, the system will be loosely coupled, so that each node accounts for a major portion of the oncoming load to itself. In section \ref{heuristic}, we will return to a variant of this local heuristic used in FastRoute \cite{NSDI_paper} and derive analytical conditions under which the above intuition holds good.\\
In this section, we take a principled approach and propose an iterative load-management algorithm, which is provably optimal for arbitrary system-parameters ($\bm{A},\bm{C}$). In this algorithm, it is enough for each node $i$ to know its own local DNS and HTTP-request arrival rates (i.e., $A_i$ and $S_i(t)$ respectively) and the entries corresponding to the (static) $i$\textsuperscript{th} row and column of the correlation matrix $\bm{C}$ (i.e., $C_{i\cdot}, C_{\cdot i}$). No non-local knowledge of the dynamic loading conditions of other nodes $j\neq i$ is required for its operation. 

\subsection{Mathematical formulation}
Consider the following optimization problem. The decision variables $(\bm{x}, \bm{S})$ have the same interpretation as above. The cost-functions, constraints and their connection to the load-management problem are discussed in detail subsequently.
\begin{eqnarray}\label{obj_fun}
 \text{Minimize} \hspace{20pt} W(\bm{x},\bm{S})\equiv \sum_{i=1}^{N}\big(g_i(S_i) + h_i(x_i)\big) 
\end{eqnarray}
Subject to,
\begin{eqnarray} \label{load_x}
S_i=\sum_{j=1}^{N}C_{ji}A_jx_j,\hspace{10pt} \forall i=1,2,\ldots, N
\end{eqnarray}
\begin{eqnarray*}
\bm{x}\in X, 
\bm{S} \in \Sigma_T
\end{eqnarray*}
 \paragraph*{Discussion}The \underline{first component} of the cost function, $g_i(S_i)$, denotes the cost for overloading the $i$\textsuperscript{th} proxy. In our numerical work, we take $g_i(\cdot)$ to be proportional to the average aggregate queuing delay for an $M/G/1$ queue with processing-capacity $T_i$ \cite{bertsekas1987data}, i.e., 
\begin{eqnarray} \label{ex1}
g_i(S_i)&=& \begin{cases}
\frac{\eta_i S_i}{1-\frac{S_i}{T_i}}, \hspace*{20pt}\text{if   } S_i\leq T_i \\
 \infty \hspace*{10pt} \text{o.w.} 
\end{cases}
\end{eqnarray}
Here $\eta_i$ is a positive constant, denoting the relative cost 
(say, in dollars) due to per unit delay.\\
The \underline{second component} of the cost function $h_i(x_i)$ denotes the cost due to \emph{round-trip-latency} of requests routed to the secondary layer ($L_2$). As an example, in a popular model \cite{roughgarden2005selfish} the delay incurred by a single packet over a congested path varies \emph{affinely} with the offered load. Since the rate of traffic sent to the secondary layer by node $i$ is $A_i(1-x_i)$, in this model, the cost-function $h_i(x_i)$ may be written as follows  
\begin{eqnarray} \label{ex2}
h_i(x_i)=\gamma_i A_i(1-x_i)\big(d_i+A_i(1-x_i)\big)
\end{eqnarray}
where $d_i$ is a (suitably normalized) round-trip-latency parameter from the node $i$ to $L_2$ and $\gamma_i$ is a positive constant denoting the relative cost due to per unit latency. We use the cost-functions described above for our numerical work.\\ 
The \underline{constraint set} $X=[0,1]^N$ represents the $N$-dimensional unit hypercube and the set $\Sigma_T$ captures the capacity constraints of the proxies, e.g., if the proxy $i$ has capacity $T_i$ then we have  
\begin{eqnarray*}
 \Sigma_{\bm{T}}=\{\bm{S}: S_i \leq T_i, \forall i=1,2,\ldots, N\}
\end{eqnarray*}
In general, the functions $g_i(\cdot), h_i(\cdot)$ are required to be closed, proper and convex \cite{bertsekas1999nonlinear}. We also assume the functions $g_i(\cdot)$s to be monotonically increasing. Hence, we can replace the equality constraint \eqref{load_x} by the following inequality constraint, without loss of optimality
\begin{eqnarray*}
\sum_{j=1}^{N}C_{ji}A_jx_j \leq S_i,\hspace{10pt} \forall i=1,2,\ldots, N
\end{eqnarray*} 
This is because, if the optimal $S_i^*$ is strictly greater than the LHS, we can strictly (and feasibly) reduce the objective value by reducing $S_i^*$ to the level of LHS, resulting in contradiction.\\
Hence, the above load management problem is equivalent to the following optimization problem $\textbf{P}_1$ :\\
\textbf{Minimize} 
\hspace{10pt} $W(\bm{x},\bm{S})= \sum_{i=1}^{N}\big(g_i(S_i) + h_i(x_i)\big) $\\
\textbf{Subject to},
\begin{eqnarray}
&&\sum_{j=1}^{N}C_{ji}A_jx_j \leq S_i,\hspace{10pt} \forall i=1,2,\ldots, N \label{constr1}\\
&&\bm{x}\in X, \bm{S} \in \Sigma_{\bm{T}}
\end{eqnarray} 
Where, $X=[0,1]^N$ and $\Sigma_{\bm{T}}=\{\bm{S}: S_i \leq T_i, \forall i=1,2,\ldots, N\}$.

Since the objective functions as well as the constraint sets of the problem $\bm{P}_1$ are all convex \cite{bertsekas1999nonlinear}, we immediately have the following lemma:
\begin{lemma}
 The problem $\textbf{P}_1$ is convex. 
\end{lemma}

\subsection{The Dual Decomposition Algorithm } 
\label{dual_section}
In this section we derive a dual algorithm \cite{kelly1997charging}, \cite{lobel2011distributed}, \cite{eryilmaz2010distributed} for the problem $\textbf{P}_1$ and show how it leads to a distributed implementation with negligible control overhead.   \\
By associating non-negative dual variable $\mu_i$ to the $i$\textsuperscript{th} constraint in \eqref{constr1} for all $i$, the Lagrangian of $\textbf{P}_1$ reads as follows:

\begin{eqnarray}\label{dual_obj}
\mathcal{L}(\bm{x}, \bm{S}, \bm{\mu})= \sum_{i=1}^{N}\big( g_i(S_i)-\mu_iS_i) \big)+ \nonumber \\
\sum_{i=1}^{N}\bigg(h_i(x_i)+A_ix_i\big(\sum_{j=1}^{N}\mu_jC_{ij}\big)\bigg)
\end{eqnarray}
This leads to the following dual objective function \cite{bertsekas1999nonlinear}
\begin{eqnarray} \label{dual_opt}
D(\bm{\mu})=\inf_{\bm{x}\in X,\bm{S}\in \Sigma_{\bm{T}}} \mathcal{L}(\bm{x}, \bm{S}, \bm{\mu})
\end{eqnarray}
We now exploit the \emph{separability} property of the dual objective \eqref{dual_obj} to reduce the problem \eqref{dual_opt} into following two one-dimensional sub-problems: 
\begin{equation} \label{dual_sep}
\left.
\begin{aligned} 
S_i^*(\bm{\mu})&=\inf_{0\leq S_i\leq T_i} \bigg(g_i(S_i)-\mu_iS_i  \bigg) \quad \\
x_i^*(\bm{\mu})&=\inf_{0\leq x_i\leq 1}\bigg(h_i(x_i)+A_i\beta_i(\bm{\mu})x_i \bigg)
\end{aligned}
\right\}
\end{equation}
The scalar $\beta_i(\bm{\mu})$ is defined as 
\begin{eqnarray} \label{beta_eqn}
\beta_i(\bm{\mu})=\sum_{j=1}^{N}\mu_jC_{ij} =\bm{C}_i^{T}\bm{\mu},
\end{eqnarray}
where $\bm{C}_i$ is the $i$\textsuperscript{th} row of the correlation-matrix $\bm{C}$.\\
The factor $\beta_i(\bm{\mu})$ couples the offload decision of node $i$ with the entire network. Once the value of $\beta_i(\bm{\mu})$ is available to the node $i$, it has all the required information to \emph{locally} solve the corresponding sub-problems \eqref{dual_sep} and hence evaluate the dual objective $D(\bm{\mu})$ for a fixed $\bm{\mu}\geq \bm{0}$. These solutions may even be obtained in closed form in some cases. In sub-section \ref{dist_imp}, we will show how this factor $\beta_i(\bm{\mu})$ may be made available to each node $i$ on-the-fly. \\
 With the stated assumptions on the cost functions, there will be no duality-gap \cite{bertsekas1999nonlinear}. Convex duality theory  guarantees the existence of an optimal dual variable $\bm{\mu}^*\geq \bm{0}$ such that solution to the relaxed problem \eqref{dual_sep} corresponding to $\bm{\mu}^*$ gives an optimal solution to the original constrained optimization problem $\textbf{P}_1$. To obtain the optimal dual variable $\bm{\mu}^*$, we solve dual of the problem $\textbf{P}_1$, given as follows 
\begin{eqnarray}
\textbf{Maximize} \hspace{10pt} D(\bm{\mu}) \label{dual_prob}\\
\textbf{subject to,} \hspace{10pt}\bm{\mu} \geq \bm{0} \nonumber
\end{eqnarray}

The dual problem given in Eqn. \eqref{dual_prob} is well-known to be convex \cite{bertsekas1999nonlinear}. To solve the dual problem, we use the dual super-gradient algorithm \cite{nedic2008convex}, which will be shown to be amenable to a distributed implementation.\\
At the $k$\textsuperscript{th} step of the iteration, a super-gradient $\bm{g}(\bm{\mu}(k))$ of the dual function $D(\bm{\mu})$ at the point $\bm{\mu}=\bm{\mu}(k)$ is  given by $\partial D(\bm{\mu}(k))=\bm{S}^{\text{obs}}(k)-\bm{S}(k)$ \cite{bertsekas1999nonlinear}, where $S_i^{\text{obs}}(k)$ is the observed rate of arrival of incoming traffic at proxy $i$, i.e., 
 \begin{eqnarray}
 S_i^{\text{obs}}(k)\equiv \sum_{j=1}^{N}C_{ji}A_jx_j^*(k),
 \end{eqnarray}
 and $x_i^*(k)$ and $S_i^*(k)$ are the primal variables obtained from Eqn. \eqref{dual_sep}, evaluated at the current dual variable $\bm{\mu}=\bm{\mu}(k)$. Following a super-gradient step, the dual variables $\bm{\mu}(k)$ are iteratively updated component-wise at each node $i$ as follows: 
\begin{eqnarray} \label{sub_grad_step_1}
\mu_i(k+1) = \bigg(\mu_i(k)+ \alpha \big(S_i^{\text{obs}}(k)-S_i^*(k)\big)\bigg)^+
\end{eqnarray}
Here $\alpha$ is a small positive step-size constant, whose appropriate value will be given in Theorem \eqref{bertsekas_opt}. Since the problem-parameters slowly vary over time, a stationary algorithm is practically preferable. Hence, we used a constant step-size $\alpha$, rather than a sequence of diminishing step-sizes $\{\alpha_k\}$.  The above constitutes theoretical underpinning of steps (5) and (6) of the distributed algorithm given in Algorithm \ref{algo_dual}. \\
\subsection{Convergence of the Dual Algorithm}
To prove the convergence of the above algorithm, we first \emph{uniformly} bound the $\ell_2$ norm of the super-gradients $\bm{g}(\bm{\mu}(k))$:
\begin{eqnarray*}
\bm{g}(\bm{\mu}(k))\equiv \bm{S}^{\text{obs}}(k)-\bm{S}(k)=\bigg(\sum_{j=1}^{N}C_{ji}A_jx_j^*(k)-S_i^*(k)\bigg)_{i=1}^{N}
\end{eqnarray*} 
We start with the following lemma. 
\begin{lemma} \label{subgrad_bounded}
If the total external DNS-query arrival rate to the system is bounded by $A_{\max}$ (i.e. $\sum_i A_i \leq A_{\max}$) and the maximum processing-capacity of individual proxies is bounded by $T_{\max}$ (i.e. $T_i \leq T_{\max}, \forall i$) then, for all $k\geq 1$
\begin{eqnarray}
||\bm{g}(\bm{\mu}(k))||_2^2\leq A_{\max}^2 + NT_{\max}^2
\end{eqnarray}
\end{lemma} 
\begin{proof}
See Appendix \ref{pf_sg_bd}. 
\end{proof}

Upon bounding the super-gradients uniformly for all $k$, the convergence of the dual algorithm follows directly from Proposition 2.2.2 and 2.2.3 of \cite{bertsekas2015convex}. In particular, we have the following theorem:
\begin{theorem} \label{bertsekas_opt}
For a given $\epsilon>0$, let the step-size $\alpha$ in Eqn. \eqref{sub_grad_step_1} be chosen as 
$\alpha = \frac{2\epsilon}{A_{\max}^2+NT^2_{\max}} $. 
Then, 
 \begin{itemize}
 \item The sequence of solutions produced by the dual algorithm described above converges within an $\epsilon$-neighbourhood the optimal objective value of the problem $\textbf{P}_1$.
 \item The rate of convergence of the algorithm to the $\epsilon$-neighborhood of the optima after $k$-steps is given by $c/\sqrt{k}$ where $c \sim \Theta(\sqrt{N})$.
 \end{itemize}
\end{theorem}
The above result states that the rate of convergence of the dual algorithm decreases roughly at the rate of $\Theta(\sqrt{N})$, where $N$ is the total number of nodes in the system. This is expected as more nodes in the system would warrant greater amount of inter-node coordination to converge to the optima.

\subsection{\textsc{FastControl} Packets and Distributed Implementation} \label{dist_imp}
In the previous subsection, we derived a provably optimal load-management algorithm, which is implementable in practice, provided each node $i$ knows how to obtain the value of $\beta_i(\bm{\mu}(k))$ in a decentralized fashion. To accomplish this, we now introduce the novel idea of \textsc{FastControl} packets.  In brief, it exploits the underlying anycast architecture of the system (Eqn.\eqref{S_eqn}) for \emph{in-network computation} of the coupling factor $\beta_i(\bm{\mu}(k))$  (Eqn. \eqref{beta_eqn}) for all $i$. \\
\textsc{FastControl} packets are special-purpose control packets (different from the regular data-packets), each belonging to any one of the $N$ distinct \emph{categories}. The category of each \textsc{FastControl} packet is encoded in its packet-header. These packets are generated in a controlled manner by using a javascript embedded in responses to user DNS-requests (similar to how data was generated to calculate the $C$ matrix offline \cite{NSDI_paper}).  The javascript forces users to download a small image from a URL that is not affected by the load management algorithm.  
 DNS-servers in each node are configured to respond back with anycast IP address for the primary layer (i.e. $\mathcal{I}_1$) for this special DNS-query.  The use of various categories of FastControl packets will be clear from the description of the following distributed protocol used for determining $\beta_i(\bm{\mu}(k))$:\\
\begin{itemize} 
\item  At step $k$, each node $i$ \textbf{forces} generation of \textsc{FastControl} packets of category $j$ (through its response to DNS-queries) at the rate 
\begin{eqnarray} \label{rate_gen}
r_{ij}(k)=\gamma\mu_i(k) \frac{C_{ji}}{C_{ij}}, \hspace{10pt} j=1,2,\ldots, N
\end{eqnarray}
Note that this is locally implementable, since the value of the dual variable $\mu_i(k)$ is locally available at each node $i$. Here $\gamma>0$ is a fixed system parameter, indicating the rate of control packet generation. 
\item At each step $k$, each node $i$ also \textbf{monitors} the rate of reception of \textsc{FastControl} packet of category $i$, denoted by $R_i(k)$. Using equation \eqref{rate_gen}, the total rate of reception $R_i(k)$ of $i$\textsuperscript{th} category \textsc{FastControl} packets at node $i$ is obtained as follows  
\begin{eqnarray*}
R_i(k)&=&\sum_{j=1}^{N}r_{ji}(k)C_{ji}= \sum_{j=1}^{N} \gamma \mu_j(k) \frac{C_{ij}}{C_{ji}} C_{ji}\\
&=& \gamma \beta_i(\bm{\mu}(k))
\end{eqnarray*}
Thus, 
\begin{eqnarray}
\beta_i(\bm{\mu}(k))= \frac{1}{\gamma}R_i(k)
\end{eqnarray}
\end{itemize}
Hence, the value of $\beta_i(\bm{\mu}(k))$ at node $i$ can be obtained locally by monitoring the rate of receptions of \textsc{FastControl} packets at the collocated proxy.
  A complete pseudocode of the algorithm is provided below. See Figure (\ref{dual_fig}) for a schematic diagram of a node implementing the algorithm. 
 \begin{figure}
 \centering
	\begin{overpic}[scale=0.67]{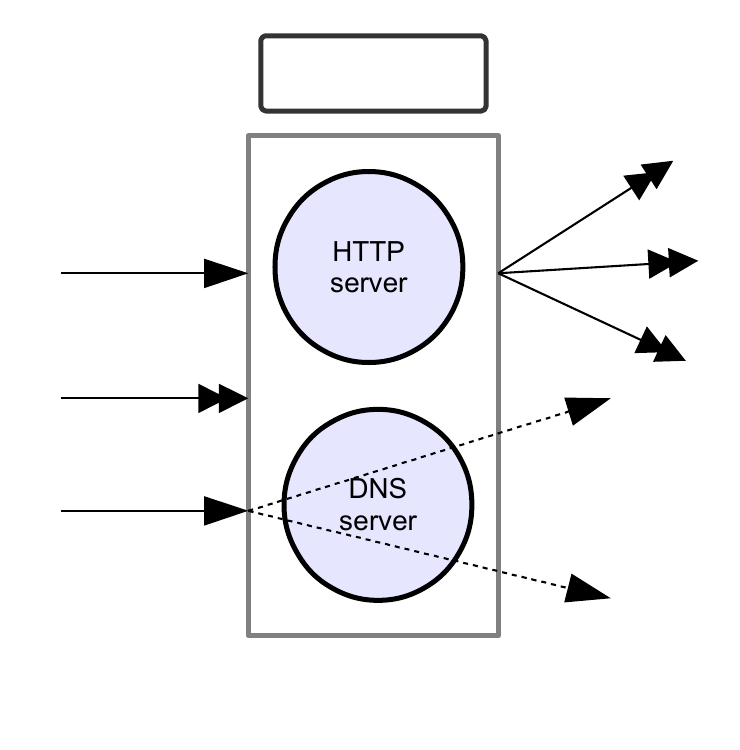}
	\put(13,68){\small{$S_i^{\text{obs}}(k)$}}
	\put(13,35){\small{$A_i$}}
	\put(-11,55){\small{(From} }
	\put(-5,50){\small{$L_1)$}}
	\put(5,54){\small{$\bigg\{$}}
	\put(13,50){\small{$R_i(k)$}}
	\put(-11,35){\small{(From} }
	\put(-11,28){\small{LDNS)}}
	\put(36,89){\small{$\mu_i(k),$}}
	\put(51,89){\small{$S_i(k)$}}
	\put(8,100){\small{\textbf{Inputs}}}
	\put(33,100){\small{\textbf{State Variables}}}
	\put(83,100){\small{\textbf{Outputs}}}
	\put(108,72){\small{\emph{FastControl}}}
	\put(115,66){\small{Pkts}}
	\put(112,59){\small{(to $L_1$)}}
	\put(89,68){\small{$\vdots$}}
	\put(89,56){\small{$\vdots$}}
	\put(92,67){\small{$r_{ij}(k)$}}
	\put(110,33){\small{Data Pkts}}
	\put(82,44){\small{$x_i(k)A_i$}}
	\put(82,37){\small{(to $L_1$)}}
	\put(81,20){\small{$(1-x_i(k))A_i$}}
	\put(86,12){\small{(to $L_2$)}}
	\put(105,66){\small{$\bigg\}$}}
	\put(105,33){\small{\bigg\}}}
	\put(42,4){\small{Node $i$}}
	\end{overpic}
\caption{\small{Node-$i$ implementing the dual algorithm}}
\label{dual_fig}
\end{figure}
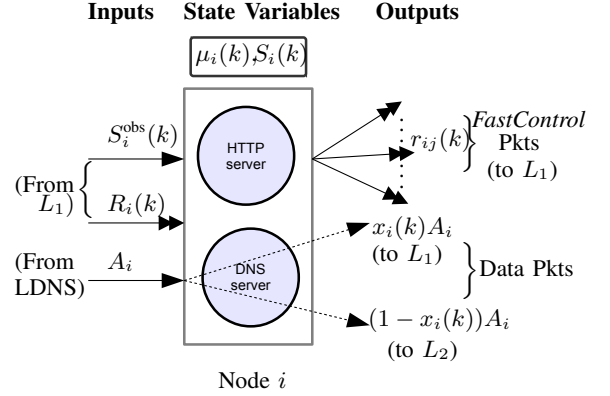
\begin{algorithm}[H]
 \small
    \begin{algorithmic}[1]
   
      \STATE \emph{Initialize:}  $\mu_i(0)\gets 0$
      \FOR {$k=1,2,3,\ldots $}
         \STATE \textbf{Monitor} $A_i(k), S^{\text{obs}}_i(k), R_i(k)$;
         \STATE Set $\beta_i(k)\gets \frac{1}{\gamma}R_i(k) ;$
         \STATE \emph{Update Primal variables:}
          \begin{eqnarray*}
           S_i(k)&\gets&\inf_{0\leq S_i\leq T_i} \big(g_i(S_i)-\mu_i(k)S_i\big)\\
           x_i(k)&\gets &\inf_{0\leq x_i\leq 1}\big(h_i(x_i)+A_i(k)\beta_i(k)x_i\big)
          \end{eqnarray*}
          \STATE \emph{Update Dual variable:}
          \begin{eqnarray*}
           \mu_i(k+1) \gets \big(\mu_i(k)+ \alpha (S_i^{\text{obs}}(k)-S_i(k))\big)^+
          \end{eqnarray*}

            \STATE Via DNS-response, force users to \textbf{generate} \textsc{FastControl} packets of category $j$, destined to $L_1$, at the rate
            \begin{eqnarray*} \label{generation}
             r_{ij}(k+1)= \gamma\mu_i(k+1) \frac{C_{ji}}{C_{ij}}, \hspace{10pt} \forall j=1,2,\ldots, N
            \end{eqnarray*}
         \STATE  For an incoming DNS-query, \textbf{respond} with the anycast IP-address for $L_1$ with probability $x_i(k)$ and IP-address for $L_2$ with probability $(1-x_i(k))$.
      \ENDFOR
    \end{algorithmic}
    \caption{Distributed Dual Decomposition Algorithm Running at Node $i$}
    \label{algo_dual}
    \end{algorithm}

\section{The Greedy Load-Management Heuristic} \label{heuristic}
In this section, we focus exclusively on the distributed load management heuristic implemented in FastRoute \cite{NSDI_paper}, a commercial layered-CDN. This heuristic ignores inter-node correlations altogether. Thus, when an individual proxy becomes overloaded, the co-located DNS-server modifies its DNS-response to redirect more traffic to the data-centers ($L_2$) and vice versa. This simple mechanism is reported to work well in practice when there is high correlation (60-80\%) between the node receiving the DNS-query and the node receiving the corresponding HTTP-request. However, in the event of sudden bursts of traffic, e.g., \emph{Flash Crowds}, this greedy heuristic often leads to an uncontrollable overload situation which necessitates manual intervention \cite{NSDI_paper}. 
%
%
%

\textbf{Algorithmic challenges:} With only local information available at each DNS-server, it faces the following dilemma: offload too little to $L_2$ and the collocated proxy, if overloaded, remains overloaded; offload too much and the users are directed to remote data-centers and receive an unnecessarily delayed response, due to high round-trip latency. The coupling among the nodes because of inter-node correlation makes this problem highly challenging and gives rise to the so-called uncontrollable overload, discussed earlier in section \ref{u_o}. 



 A pseudo-code showing the general structure of FastRoute's heuristic, running at a node $i$, is provided below. An explicit control-law modeling the heuristic will be given in section \ref{analysis_greedy}. 
%

\begin{algorithm} 
\small
    \begin{algorithmic}[1]
\FOR{ $t=1,2,3,\ldots$ }
\IF{ the $i$\textsuperscript{th} proxy is \textbf{under loaded} ($S_i(t) \leq T_i$) }
\STATE \indent \indent \textbf{increase} $x_i(t)$  proportional to $-(S_i(t)-T_i)$
\ELSE 
\STATE \indent \indent \textbf{decrease} $x_i(t)$ proportional to ($S_i(t)-T_i$)
\ENDIF
\ENDFOR
 \caption{Decentralized greedy load-management heuristic used in \emph{FastRoute} \cite{NSDI_paper}, running at the node $i$}
 \label{algo_greedy}
\end{algorithmic}
\end{algorithm}

\emph{Motivation for analyzing the heuristic:} The optimal algorithm of section-\ref{optimization} requires the knowledge of the correlation matrix $\bm{C}$ and needs to utilize additional \textsc{FastControl} packets for its operation. In this section we analyze performance of the greedy heuristic, currently implemented in \emph{FastRoute} \cite{NSDI_paper}, which does not have these implementation complexities. Since this heuristic completely ignores the inter-node correlations (given by the matrix $\bm{C}$), it cannot be expected to achieve the optima of the problem $\bm{P}_1$, in general. Instead, we measure its performance by a coarser performance-metric, given by the number of proxies that undergo uncontrollable overload condition (refer to section \ref{u_o}) under its action. Depending on target applications, this metric is often practically good enough for gauging the performance of CDNs.

 \subsection{Analysis of the Greedy Heuristic} \label{analysis_greedy}
 As before, let $x_i(t)$ denote the probability that an incoming DNS-query to the node $i$ at time $t$ is returned with the anycast address of the primary layer $L_1$. Hence, the total rate of incoming load to the proxy $i$ at time $t$ is given by,
\begin{eqnarray} \label{lin}
S_i(t)=\sum_{j=1}^{N} C_{ji}A_j x_j(t) , \hspace{10pt} i=1,2,3,\ldots,N
\end{eqnarray}
The above system of linear equations can be compactly written as follows
\begin{eqnarray}
\bm{S}(t)=\bm{B}\bm{x}(t)
\end{eqnarray}
Where, 
\begin{eqnarray}
\bm{B}\equiv \bm{C}^T\textbf{diag}(\bm{A}).
\end{eqnarray} 

Let the vector $\bm{T}$ denote the processing-capacities (thresholds) of the proxies. As described above, FastRoute's greedy heuristic \eqref{algo_greedy} monitors the overload-metric $S_i(t)-T_i$ and if it is positive (i.e., the node $i$ overloaded), it reduces $x_i(t)$ (i.e., $\frac{dx_i(t)}{dt}<0$) and if the overload-metric is negative (i.e., the node $i$ under-loaded), it increases $x_i(t)$ (i.e., $\frac{dx_i(t)}{dt}>0$) proportional to the overload. We consider the following explicit control-law complying with the above general principle:
\begin{eqnarray} \label{ode}
\frac{dx_i(t)}{dt} = -\beta R(x_i(t))\big(\bm{B}\bm{x}(t)-\bm{T}\big)_i, \hspace{10pt} \forall i
\end{eqnarray}
The factor $R(x_i(t))\equiv x_i(t)(1-x_i(t))$ is a non-negative \emph{damping} component, having the property that $R(0)=R(1)=0$. This non-linear factor is responsible for restricting the trajectory of $\bm{x}(t)$ to the $N$-dimensional unit hypercube $\bm{0}\leq \bm{x}(t) \leq \bm{1}$, ensuring the feasibility of the control \eqref{ode}\footnote{Remember that $x_i(t)$'s, being probabilities, must satisfy $0\leq x_i(t)\leq 1, \forall t, \forall i$}. The scalar $\beta>0$ is a sensitivity parameter, relating the robustness of the control-strategy to the local observations at the nodes.   \\
The following theorem establishes soundness of the control \eqref{ode}:
\begin{theorem} \label{existence_uniqueness_thm}
Consider the following system of ODE
\begin{eqnarray} \label{eq1}
\dot{x_i}(t)= -R(x_i(t))(\bm{B}\bm{x}(t)-T)_i, \hspace{10pt} \forall i
\end{eqnarray}
where $R:[0,1]\to \mathbb{R}_+$ is any $\mathcal{C}^{1}$ function, satisfying  $R(0)=R(1)=0$. \\
Let $\bm{x}(0)\in \text{int}({\mathcal{H}})$, where $\mathcal{H}$ is the $N$-dimensional unit hypercube $[0,1]^{N}$. Then the system \eqref{eq1} admits a unique solution $\bm{x}(t) \in \mathcal{C}^1$ such that $\bm{x}(t) \in \mathcal{H}, \forall t\geq 0$.
\end{theorem}
\begin{proof}
 See Appendix \ref{existence_uniqueness_proof}.
\end{proof}
The following theorem reveals an interesting feature of the greedy algorithm, which states that, along any periodic trajectory the average load at any node $i$ is equal to the threshold $T_i$ of that node, and hence they are stable \emph{on the average}. 
\begin{theorem} \label{limit_cycle}
Consider the system \eqref{ode} with possibly time-varying DNS-influenced arrival rate vector $\bm{A}(t)$ such that the system operates in a periodic orbit. Then the time-averaged http-load on any node $i$ is equal to the threshold $T_i$ of that node.  
\end{theorem} 
\begin{proof}
See Appendix \ref{limit_cycle_proof}.  
\end{proof}

\subsection{Avoiding Locally Uncontrollable Overload}
Having established the feasibility and soundness of the control-law \eqref{ode}, we return to the original problem of locally uncontrollable overload, described in section \ref{u_o}. In the following, we derive sufficient conditions for the correlation matrix $\bm{C}$ and the external arrival rate $\bm{A}$, for which the system is stable, in the sense that no locally uncontrollable overload situation takes place. 

\subsection*{Characterization of the Stability Region}
For a fixed correlation matrix $\bm{C}$, we show that if the arrival rate vector $\bm{A}$ lies within a certain polytope $\bm{\Pi}_{\bm{C}}$, the system is stable in the above sense, under the action of the greedy load-management heuristic. The formal derivation of the result is provided in Appendix \ref{stability_region}, which involves linearization of the ODE \eqref{ode} around certain fixed points. Here we outline a simple and intuitive derivation of the stability region $\bm{\Pi}_{\bm{C}}$. \\
We proceed by contradiction. Suppose that node $i$ is facing a locally uncontrollable overload at time $t$. Hence, by definition, the following two conditions must be satisfied at node $i$
\begin{eqnarray}
 S_i(\infty)-T_i >0 \label{ol}, \hspace{5pt}\text{and} \hspace{5pt}
 x_i(\infty)=0 \label{uc}
\end{eqnarray}
Here Eqn. \eqref{ol} denotes the fact that FastRoute node $i$ is \emph{overloaded}, i.e., the incoming traffic to node $i$'s proxy is more than the capacity of the node $i$. Eqn. \eqref{uc} denotes the fact that this overload is \emph{locally uncontrollable}, since even after node $i$'s DNS-server has offloaded \emph{all} incoming DNS-influenced arrivals to $L_2$ (the best that it can do with its local information), it is facing the overload situation. The above two equations imply that the following condition holds at the node $i$:
\begin{eqnarray} \label{cond1}
 \sum_{j \neq i} C_{ji}A_jx_j(\infty) > T_i,
\end{eqnarray}
where we have used Eqn. \eqref{lin} and the fact that $x_i(\infty)=0$. Since $0\leq x_j(\infty) \leq 1 $,
a necessary condition for uncontrollable overload \eqref{cond1} at node $i$ is $\sum_{j \neq i} C_{ji}A_j > T_i$. Thus, if $\sum_{j \neq i} C_{ji}A_j \leq  T_i$, then the locally uncontrollable overload is avoided at the node $i$ by the greedy heuristic. Taking into account all nodes, we see that if the external DNS-query arrival rate $\bm{A}$ lies in the polytope $\bm{\Pi}_{\bm{C}}$ defined as  
\begin{eqnarray} \label{polytope}
 \bm{\Pi}_{\bm{C}}=\{\bm{A}\geq \bm{0}: \sum_{j \neq i} C_{ji}A_j \leq  T_i, \hspace{5pt} \forall i=1,2,\ldots,N\}
\end{eqnarray}
then the locally uncontrollable overload situation is avoided at \emph{every} node and the system is stable. Somewhat surprisingly, by exploiting the exact form of the control-law \eqref{ode}, we also show that a two-node system (as depicted in Figure \ref{two-node_fig}) is able to control DNS-load $\bm{A}$ of any magnitude, under certain favorable conditions on the correlation matrix $\bm{C}$.


\subsubsection*{Special Case}[\textbf{Two-node System}]

Consider a two-node CDN discussed earlier in Section \ref{optimization} (see Figure \ref{two-node_fig}). Let the correlation matrix $\bm{C}$ for the system be parametrized as follows:
\begin{eqnarray}
\bm{C}(\alpha, \beta)=
\begin{pmatrix}
\alpha && 1-\alpha \\ 1-\beta && \beta 
\end{pmatrix}
\end{eqnarray}
Then we have the following theorem :
\begin{theorem} \label{two_node_lemma}
1) The system does not possess \emph{any} periodic orbit for \emph{any} values of its defining parameters: $\bm{A},\bm{C}(\alpha,\beta),\bm{T}$. Thus the system never oscillates.\\
2) If $\alpha > \frac{1}{2}$ and $\beta > \frac{1}{2}$ then the system is locally controllable (i.e., no locally uncontrollable overload) for \emph{all} arrival rate-pairs $(A_1,A_2)$.\\
3) If $\alpha < \frac{1}{2}$ and $\beta < \frac{1}{2}$ then a sufficient condition for local controllability of the system is $A_1<\frac{T_1}{1-\alpha}$, $A_2 < \frac{T_2}{1-\beta}$. 
\end{theorem}
\begin{proof}
The proof of part-(1) follows from Dulac's criterion \cite{strogatz2014nonlinear}, whereas proof for part-(2) and (3) follows from linearization arguments. See Appendix \ref{two_node_lemma_proof} for details.
\end{proof}
We emphasize that the part-(2) of the Theorem \ref{two_node_lemma} is surprising, as it shows that the system remains locally controllable, no matter how large the incoming DNS-query arrival rate be (c.f. Section \ref{u_o}) . The 2D vector-field plot in Figure \ref{vector_field_plot} illustrates the above result. In Figure 4(a), the matrix $\bm{C}$ is taken to be one satisfying the condition of part (2) of lemma \ref{two_node_lemma}. As shown, all four phase-plane trajectories with different initial conditions converge to an interior fixed point $(x_1(\infty)>0, x_2(\infty)>0)$. \\
For the purpose of comparison, in Figure 4(b) we plot the 2D vector-field  of a locally uncontrollable system (e.g., the system in Fig. \ref{two-node_fig}). It is observed that all four previous trajectories converge to the uncontrollable attractor $(x_1(\infty)=1,x_2(\infty)=0)$. From the vector-field plot, it is also intuitively clear why a periodic orbit can not exist in the system.

 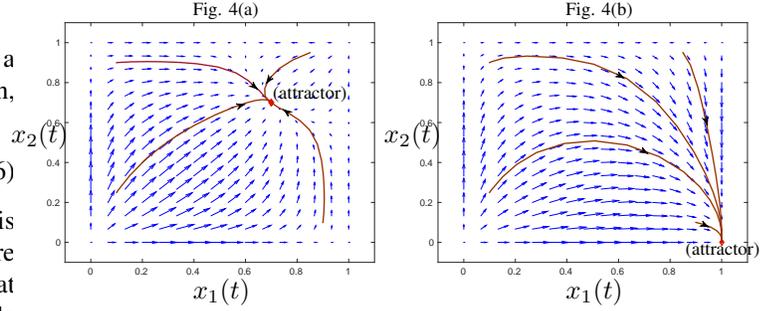
\begin{figure}
 \begin{minipage}[b]{0.53\linewidth}
 \centering 
 \hspace*{-15pt} 
     \begin{overpic}[scale=0.25]{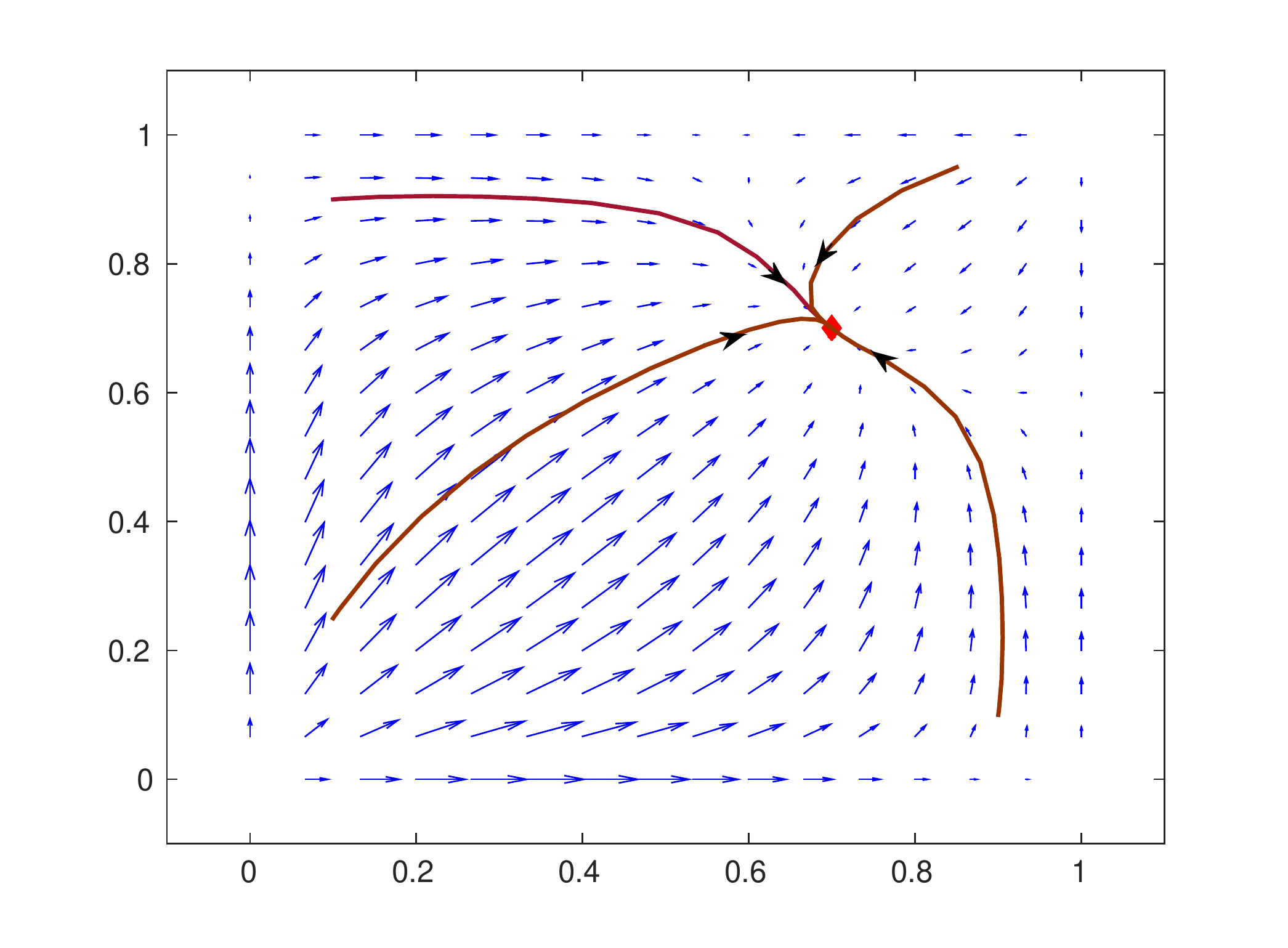}
     \put(45,-1){$\tiny{x_1(t)}$}
     \put(-0.5,38){$\tiny{x_2(t)}$}
      \put(65,49){\text{\scriptsize{(attractor)}}}
       \put(45,70){\text{\scriptsize{Fig. 4(a)}}}
    \end{overpic}
 \end{minipage}
 \begin{minipage}[b]{0.46\linewidth} 
 \centering
 \hspace*{-10pt}
  \begin{overpic}[scale=0.25]{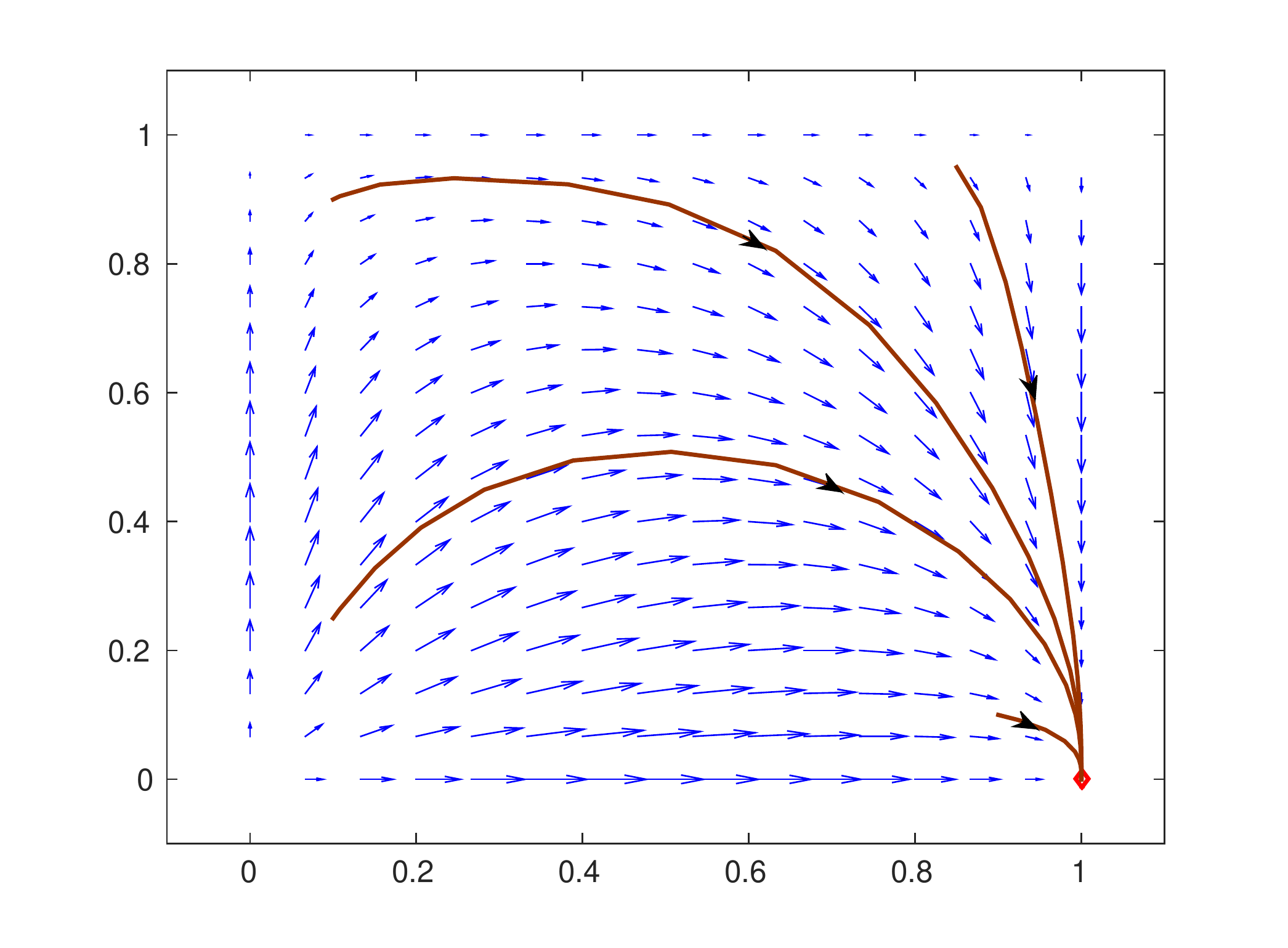}
 \put(45,-1){$\tiny{x_1(t)}$}
 \put(-0.5,38){$\tiny{x_2(t)}$}
 \put(45,70){\text{\scriptsize{Fig. 4(b)}}}
 \put(75,10){\text{\scriptsize{(attractor)}}}
   \end{overpic}
   \end{minipage}

\caption{2D vector-fields of a two-node system illustrating locally controllable (Fig. $4(a)$) and locally uncontrollable (Fig. $4(b)$) overloads.}
\label{vector_field_plot}
\end{figure}
\paragraph*{Application of Theorem \ref{two_node_lemma}}
Consider a setting where, instead of making locally greedy offloading decisions, nodes are permitted to partially coordinate their actions. Also assume that there exists a partition of the set of nodes into two non-empty disjoint sets $G_1$ and $G_2$, such that their effective self-correlation values $\alpha(G_1)$ and $\beta(G_2)$, defined by 
\begin{eqnarray*}
\alpha(G_1)=\frac{\sum_{i \in G_1}\sum_{j \in G_1}C_{ij}}{|G_1|}, \beta(G_2)=\frac{\sum_{i \in G_2}\sum_{j \in G_2}C_{ij}}{|G_2|}
\end{eqnarray*}
satisfy the condition (2) of Theorem \ref{two_node_lemma}, i.e. $\alpha(G_1)>\frac{1}{2}, \beta(G_2)> \frac{1}{2}$. Then if the nodes in the sets $G_1$ and $G_2$ coordinate and jointly implement the greedy policy, then the system is locally controllable for all symmetric arrivals.  




\section{Numerical Evaluations}\label{simulations}
We use the operational FastRoute CDN to identify relevant system parameters for critical evaluations of the optimal algorithm and the heuristic. Currently, FastRoute has $N=48$ operational nodes, spreading throughout the world \cite{NSDI_paper}. The inter-node correlation matrix $\bm{C}$ is computed using system-traces spanning over three months.\\
For our performance evaluations, we use the cost-functions given in Eqns. \eqref{ex1} and \eqref{ex2}. Different (normalized) parameters appearing in the cost-functions are chosen as follows 
\begin{eqnarray}
\gamma_i=10, \eta_i=1, T_i=0.7, d_i \sim U_i[0,1]\hspace{25pt}\forall i
\end{eqnarray}
where $U_i[0,1]$'s denote i.i.d. uniformly distributed random variables in the range $[0,1]$. The DNS-query arrival rates $A_i$'s are assumed to be i.i.d. distributed  according to a Poisson variable with mean $\bar{A}$, which varies across the range $[0.1,10]$. \\
The optimal solutions of the Lagrangian relaxations in Eqns. \eqref{dual_sep} for the above cost-functions can be obtained in closed form as follows :
\begin{eqnarray}
S_i^*(\bm{\mu})=T_i\max \bigg(0, 1-\sqrt{\frac{\eta_i}{\mu_i}}\bigg)
\end{eqnarray}
and,
\begin{eqnarray}
x_i^*(\bm{\mu})=\begin{cases}
1, \hspace{15pt}\text{if  }c_{2i}>0\\
1+\frac{c_{2i}}{2c_{1i}}, \text{  if  } c_{2i}\leq 0 \text{ and }2c_{1i}\geq -c_{2i}\\
0, \hspace{15pt}\text{o.w.}
\end{cases}
\end{eqnarray}
where $c_{1i}\equiv A_i\gamma_i$ and $c_{2i}\equiv \gamma_id_i-\beta_i(\bm{\mu})$.\\
For each value of $\bar{A}$, we run the simulation $N_E=100$ times, by randomizing over both $A_i$ and $d_i, \forall i$. The mean and standard-deviation of the resulting optimal cost values are plotted in the Figure \ref{cost_function_opt}(a). As expected, the average cost increases as the overall DNS-query arrivals to the system increases. However, the resulting cost remains finite always. This implies that \emph{none} of the proxies are overloaded. \\
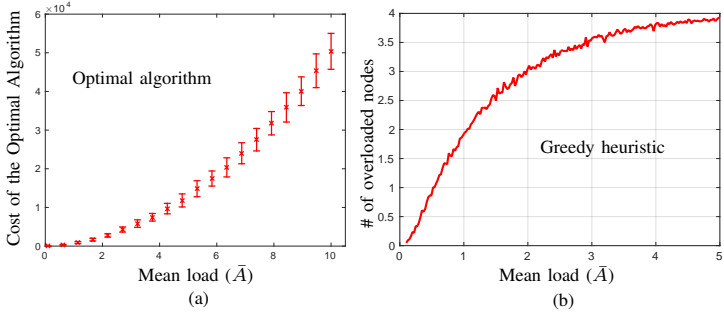
\begin{figure}
 \begin{minipage}[b]{0.54\linewidth}
 \centering
 \hspace*{-37pt} 
     \begin{overpic}[scale=0.242]{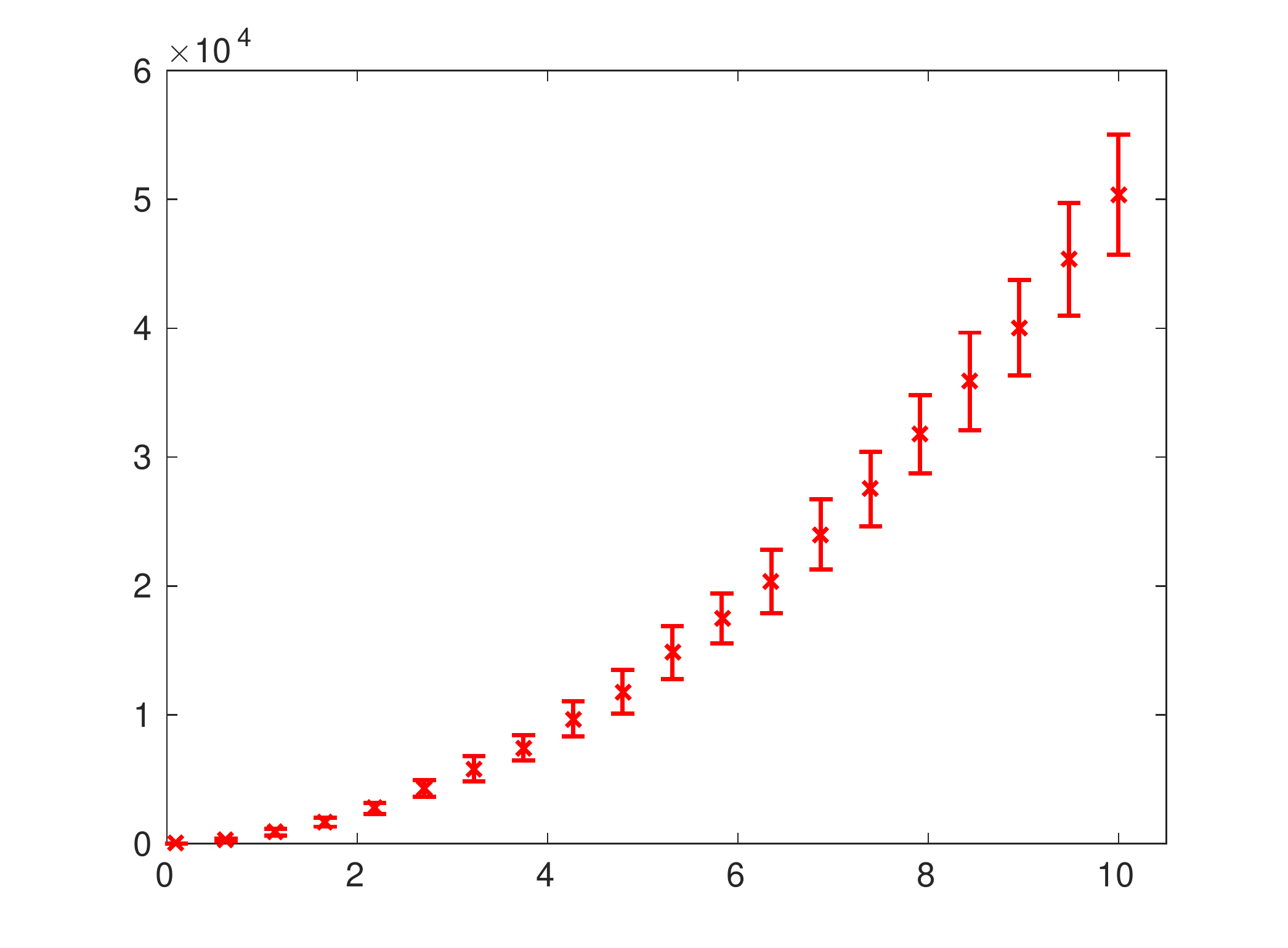}
     \put(37,-1){\scriptsize{Mean load ($\bar{A}$)}}
     \put(50,-7){\scriptsize{(a)}}
    \put(20,50){\scriptsize{Optimal algorithm}}
     \put(3,8){\rotatebox{90}{\scriptsize{Cost of the Optimal Algorithm}}}
    \end{overpic}
  \end{minipage}  
   \begin{minipage}[b]{0.45\linewidth}
    \label{greedy_overload_plot}
    \centering
    \hspace*{-30pt} 
  \begin{overpic}[scale=0.20]{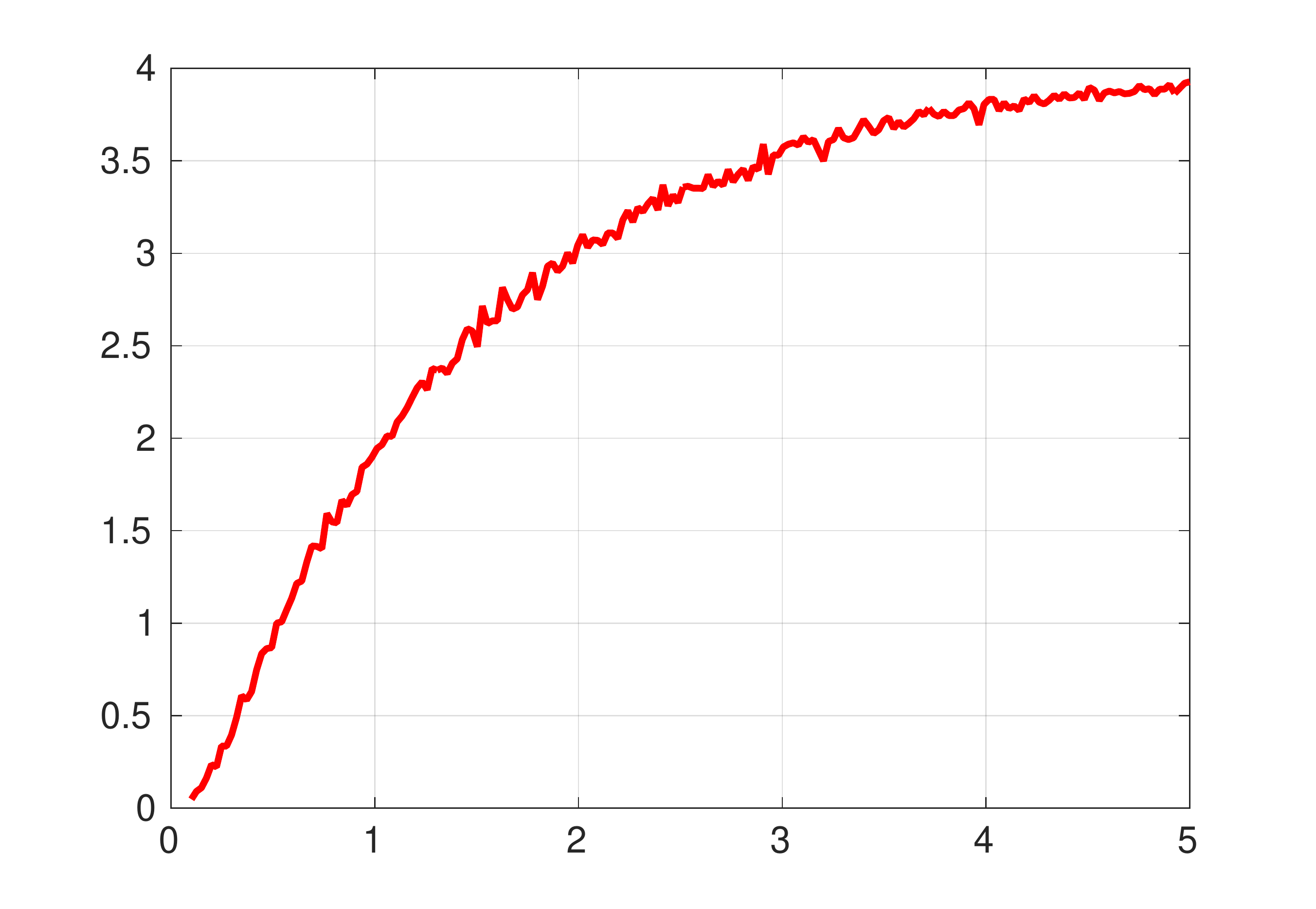}
 \put(37,-1){\scriptsize{Mean load ($\bar{A}$)}}
 \put(47,30){\scriptsize{Greedy heuristic}}
 \put(50,-7){\scriptsize{(b)}}
 \put(3.5,12){\rotatebox{90}{\scriptsize{\# of overloaded nodes}}}
   \end{overpic}
\end{minipage}
 \caption{(a) Statistical variation of the cost of the optimal algorithm with mean load (DNS-query arrival rate) $\bar{A}$. (b) Average number of nodes undergoing uncontrollable overload condition under the action of the greedy algorithm (Threshold $T_i=0.7$ for all nodes).}
\label{cost_function_opt}
\end{figure}
The above observation is in sharp contrast with the situation using the greedy-heuristic, the subject of Figure \ref{cost_function_opt}(b). Here we plot the number of overloaded proxies for different values of $\bar{A}$, keeping all other system-parameters the same. While the greedy algorithm does yield acceptable result for small values of $\bar{A}<< T_i=0.7$, we see that as many as $4$ proxies (out of a total of $48$ proxies) undergo locally uncontrollable overload for relatively large values of $\bar{A}$.  \\
Thus, depending on the computed correlation-matrix $\bm{C}$ and a projected bound of the DNS-query arrival rate $\bm{A}$, we can make an informed decision about the choice of the algorithms to employ in a CDN and the inherent complexity-vs-optimality trade-off it entails. 

\section{Conclusion}
In this paper we formalize the load management problem in modern CDNs utilizing anycast. We first formulate the problem as a convex optimization problem and study its dual and an associated algorithm. The novel idea of \textsc{FastControl} packets facilitates distributed implementation of the proposed dual algorithm. Next we analyze stability properties of a greedy heuristic, currently in operation in a major commercial CDN. We find that the optimal algorithm significantly out-performs the heuristic for moderate-to-high value of system load. However, the heuristic performs satisfactorily given that the system is \emph{loosely-coupled} and the offered load is low. Thus an informed choice between these two algorithms may be made depending on the range of system-parameters and the desired optimality/complexity trade-off for a particular CDN. Future work would involve investigating the amount of \textsc{FastControl} packets necessary for the dual algorithm to work in the presence of random packet-loss and delayed feedback. Also, It would be interesting to generalize the findings of Theorem \ref{two_node_lemma} for the case of more than two nodes.

\bibliographystyle{IEEEtran}
\bibliography{Sigcomm_submission}
\section{Appendix}
\subsection{Proof of Lemma \ref{subgrad_bounded}} \label{pf_sg_bd}
\begin{proof}
We have, 
\begin{eqnarray}
||\bm{g}_k||_2^2 &=& \sum_{i=1}^{N} (S_i^{\text{obs}}(k)-S_i(k))^2 \\
&\leq & \sum_{i=1}^{N} \big(S_i^{\text{obs}}(k)\big)^2 + \sum_{i=1}^{N}S_i^2(k) \label{non_neg}\\
&\leq & \sum_{i=1}^{N} \big(\sum_{j=1}^{N}C_{ji}A_jx_j(k)\big)^2 + NT_{\max}^2 \label{feas1}\\
&\leq & \big(\sum_{i=1}^{N}\sum_{j=1}^{N}C_{ji}A_j\big)^2 + NT_{\max}^2 \label{feas2}\\
&=& \big(\sum_{j=1}^{N}A_j \sum_{i=1}^{N}C_{ji}\big)^2 + NT_{\max}^2 \label{order} \\
& \leq & (\sum_{j=1}^{N}A_j)^2+ NT_{\max}^2 \label{corr_mat}\\
&\leq & A_{\max}^2 + NT^2_{\max}
\end{eqnarray}
Here Eqn. \eqref{non_neg} follows from non-negativity of $S_i^{\text{obs}}$ and $S_i$, Eqn. \eqref{feas1} follows from the defining equation of $S_i^{\text{obs}}$ and the constraint that $S_i \leq T_i$ (viz. Eqn. \eqref{dual_sep}), Eqn. \eqref{feas2} follows from the constraint $0\leq x_i \leq 1, \forall i$, Eqn. \eqref{order} follows from the change of order of summation and finally Eqn. \eqref{corr_mat} follows from the fact that $\bm{C}$ is a correlation matrix and hence its rows sum to unity (viz. Eqn. \eqref{corr_mat_prop}). 
\end{proof}
\subsection{Proof of Theorem \ref{existence_uniqueness_thm}}
\label{existence_uniqueness_proof}
\begin{proof} 
First we show that, any solution of the system \eqref{ode} (if exists) must lie in the unit hypercube $\mathcal{H}$. We prove it via contradiction. On the contrary to the claim, assume that for some solution of the system \eqref{ode}, there
 exists a component $x_i(\cdot)$ and a finite time $0\leq \tau < \infty$ such that $x_i(\tau) <0$. Since $x_i(\cdot)$ is continuous and $x_i(0)>0$, by \emph{intermidiate value theorem}, there must exist a time $0<t_0<\tau$ such that $x_i(t_0)=0$. Now consider the differential equation corresponding to the $i$\textsuperscript{th} component of the system \eqref{eq1}. We substitute for all other components $\{x_j(t), j\neq i\}$ on the RHS of the following equation.
 
\begin{eqnarray} \label{eq2}
 \dot{x_i}(t)= -R(x_i)(\sum_{j}B_{ij}x_j(t)-T_i), \hspace{20pt} x_i(t_0)=0
\end{eqnarray}
 Since the vector $\bm{x}(t)$ is $\mathcal{C}^{1}$ and the regularizer $R(\cdot)$ is assumed to be $\mathcal{C}^{1}$, the RHS of the equation \eqref{eq2} is $\mathcal{C}^1$. Hence \eqref{eq2}  admits a \emph{unique local solution}. However note that the following is a solution to \eqref{eq2}
 \begin{eqnarray} \label{eq3}
  x_i(t)=0, \hspace{20pt} \forall t\geq t_0
 \end{eqnarray}
 This is because $R(0)=0$. By uniqueness, \eqref{eq3} is the \emph{only} solution to \eqref{eq2}. This contradicts the fact that $x_i(\tau)<0$. Hence $\bm{x}(t) \geq \bm{0}, \forall t\geq 0$. In a similar fashion, we can also prove that $\bm{x}(t) \leq \bm{1}, \forall t \geq 0$. This proves that all solutions to \eqref{ode} must lie in the compact set $\mathcal{H}$. \\
 To complete the proof, we observe that the RHS of  the system \eqref{ode} is locally Lipschitz at each point in the compact set $\mathcal{H}$. Thus the global existence of the solution of \eqref{ode} follows directly from Theorem 2.4 of \cite{khalil1996nonlinear}. 
 
 \end{proof}
 
 \subsection{Proof of Theorem \ref{limit_cycle}} \label{limit_cycle_proof}
 \begin{proof}
  Let $\tau$ be the period of the orbit. Consider the $i$\textsuperscript{th} differential equation 
   \begin{eqnarray} 
\dot{x_i}(t)= -R(x_i)(S_i(t)-T_i)\\
\frac{dx_i}{R(x_i)}= - (S_i(t)-T_i)dt
 \end{eqnarray}
 Since $x_i(\cdot)$ belongs to the \emph{interior} of the compact set $\mathcal{H}$, $\frac{1}{R(x_i(t))}$ does not have a zero in the denominator for the enire orbit. Hence $\frac{1}{R(x_i(t))}$ is continuous and its Riemann integral exists \cite{rudin1964principles}. Integrating both sides from $0$ to $\tau$, we have 
 \begin{eqnarray} \label{eq4}
  \int_{0}^{\tau} \frac{dx_i}{R(x_i)} = \int_{0}^{\tau} (S_i(t)-T_i) dt
 \end{eqnarray}
Let $J(x_i)$ be an anti-derivative of $\frac{1}{R(x_i)}$. Hence using the fundamental theorem of calculus \cite{rudin1964principles}, we can write the LHS of \eqref{eq4} as 
\begin{eqnarray}
 J(x_i(\tau))- J(x_i(0))= \int_{0}^{\tau} S_i(t)dt - \tau T_i
 \end{eqnarray} 
Since the orbit is assumed to have a period $\tau$, we have $x_i(\tau)=x_i(0)$. Hence $ J(x_i(\tau))- J(x_i(0))=0$. Thus we have
\begin{eqnarray*}
 \bar{S_i}\equiv\frac{1}{\tau}\int_{0}^{\tau}S_i(t)dt = T_i
\end{eqnarray*}
 \end{proof}
 
 \subsection{Formal Derivation of the Stability Condition} \label{stability_region}
 Let us write the system \eqref{ode} conveniently as $\dot{\bm{x}}= \bm{F}(\bm{x})$. Consider a fixed point $\overline{x}$ of the system such that the $k$\textsuperscript{th} node faces an uncontrollable overload condition. By \emph{Hartman-Grobman} theorem \cite{strogatz2014nonlinear}, it is enough to consider linearized version of the system to determine the stability of fixed points. The first-order Taylor expansion about the fixed point $\overline{x}$ yields the following:
\begin{eqnarray*}
\dot{\bm{x}}\approx F(\bar{\bm{x}}) + \frac{\partial \bm{F}(\bm{x})}{\partial{\bm{x}}}|_{\bm{x}=\overline{\bm{x}}}(\bm{x}-\bar{\bm{x}})= \frac{\partial \bm{F}(\bm{x})}{\partial{\bm{x}}}|_{\bm{x}=\overline{\bm{x}}}(\bm{x}-\bar{\bm{x}})
\end{eqnarray*} 
Where, $\frac{\partial \bm{F}(\bm{x})}{\partial{\bm{x}}}|_{\bm{x}=\overline{\bm{x}}}$ denotes the Jacobian \cite{ogata1970modern} of the system \eqref{ode} evaluated at the fixed point $\bm{x}=\overline{\bm{x}}$ and $B_{ij}=C_{ji}A_j$. The second equation follows because $\overline{\bm{x}}$ is assumed to be a fixed point (and consequently $\bm{F}(\overline{\bm{x}})=\bm{0})$. \\
Next we proceed to explicitly compute the Jacobian of the system \eqref{ode} at a given point $\bm{x}$. Note that the $i$\textsuperscript{th} row of the system equation is given by 
\begin{eqnarray}
F_i(\bm{x})\equiv -x_i(1-x_i)(\sum_{j}B_{ij}x_j - T)
\end{eqnarray}
Thus for $i \neq j$, we have 
\begin{eqnarray} \label{off_diag}
\frac{\partial F_i}{\partial x_j}=-x_i(1-x_i)B_{ij}
\end{eqnarray}
and for $i=j$, the diagonal entry is given by
\begin{eqnarray} \label{diag}
&& \frac{\partial F_i}{\partial x_i}=-\bigg(x_i(1-x_i)B_{ii} + (1-2x_i)(\sum_{j}B_{ij}x_j - T)\bigg) \nonumber \\
&=& -\bigg(x_i(1-x_i)B_{ii} + (1-2x_i)(S_i-T)\bigg)
\end{eqnarray}
Since the node $k$ is assumed to undergo an uncontrollable overload condition, we necessarily have $x_k=0$. Hence from Eqns \eqref{off_diag} and \eqref{diag}, in the $k$\textsuperscript{th} row of the Jacobian matrix, the off-diagonal entries are all zero and the diagonal entry is $-(S_k-T)$. Hence if $S_k-T<0$, at least one eigenvalue of the Jacobian matrix at the fixed point $\overline{x}$ is strictly positive and hence the fixed point $\overline{x}$ is unstable. This implies that a sufficient condition to avoid uncontrollable overload at node $k$ is given by 
\begin{eqnarray} \label{node_k_uo}
\sum_{j\neq k}C_{jk}A_j \leq T_k
\end{eqnarray}
where we have used the fact that $x_i(t)\leq 1, \forall i$ and $x_k=0$. The derivation of the sufficient condtion for absence of uncontrollable overload condition is completed by taking intersection of hyperplanes \eqref{node_k_uo} for all $k=1,2,\ldots, N. \hspace{5pt} \blacksquare$  
\subsection{Proof of Theorem \ref{two_node_lemma}} \label{two_node_lemma_proof}
\begin{proof}
\textbf{part-(1) [non-existence of periodic orbits]}\\
We use Dulac's criterion to prove the non-existence of periodic orbits for the general two-node system. For ease of reference, we recall Dulac's criterion below :
\begin{theorem}[Dulac's criterion \cite{strogatz2014nonlinear}]
Let $\dot{\bm{x}}=\bm{f}(\bm{x})$ be a continuously differentiable vector-field defined on a simply connected subset $R$ of the plane. If there exists a continuously differentiable, real-valued function $g(\bm{x})$ such that $\nabla \cdot (g \dot{\bm{x}})$ has one sign throughout $\mathcal{D}$, then there are no closed orbits lying entirely in $\mathcal{D}$.
\end{theorem}
Now we return to the proof of the result. Let $\mathcal{H}_2=[0,1]^2$ be the unit square, where by virtue of theorem \ref{existence_uniqueness_thm}, the trajectory of the two-node system lies for all time $t\geq 0$. Thus, it suffices to show that there does not exist any periodic orbit in its interior $\mathring{\mathcal{H}_2}\equiv \mathcal{D}$. It is obvious that the region $\mathcal{D}$ is simply connected. Now consider the following $g(\bm{x})$ for application of the Dulac's criterion, 
\begin{eqnarray}
g(\bm{x})=\frac{1}{x_1x_2(1-x_1)(1-x_2)}
\end{eqnarray}
It is easy to verify that $g(\bm{x})$ is continuously differentiable throughout the region $\mathcal{D}$. We next evaluate the divergence 
\begin{eqnarray}
\nabla \cdot (g\dot{\bm{x}})= -\beta\bigg(\frac{B_{11}}{x_2(1-x_2)}+\frac{B_{22}}{x_1(1-x_1)}\bigg)
\end{eqnarray} 
It is easy to verify that the term within the parenthesis is always strictly positive throughout the region $\mathcal{D}$. Hence, by Dulac's criterion, there are no closed orbits in $\mathcal{D}$. This proves the result.  \\
\textbf{part-(2) and (3) [controllability of the system]}\\
Let the arrival rates to node 1 and 2 be given by $A_1$ and $A_2$. Let $x_1$ and $x_2$ denote the operating point of the system at the steady-state. Our objective is to find sufficient conditions on the arrival rate vector $(A_1,A_2)$, under which the operating points $(x_1,x_2)$ such that either $x_1=0$ or $x_2=0$ (i.e. full offload to Layer-II) are avoided in the \emph{steady-state}. This will ensure that no uncontrollable overload situation takes place in the system. First we consider the fixed point 
\begin{eqnarray}
x_1=0, x_2=1
\end{eqnarray}
This fixed point will be stable if both the eigenvalues of the Jacobian matrix at this point be negative. From Eqns. \eqref{off_diag} and \eqref{diag} we have the following two conditions: 
\begin{eqnarray*}
S_1>T, S_2<T
\end{eqnarray*}
i.e.,
\begin{eqnarray*}
(1-\alpha)A_2 > T, \beta A_2 <T
\end{eqnarray*}
i.e., 
\begin{eqnarray} \label{st_1}
\frac{T}{1-\alpha}<A_2< \frac{T}{\beta} 
\end{eqnarray}
Similarly, analyzing the stability of the fixed points around the point $x_1=1, x_2=0$, we obtain that this fixed point will be unstable if 
\begin{eqnarray*}
S_1<T, S_2>T
\end{eqnarray*}
i.e.,
\begin{eqnarray*}
\alpha A_1 < T, (1-\beta)A_1 > T
\end{eqnarray*}
i.e.,
\begin{eqnarray} \label{st_2}
\frac{T}{1-\beta}< A_1 < \frac{T}{\alpha} 
\end{eqnarray}
Note that if $\alpha + \beta >1 $, both the above regions \eqref{st_1} and \eqref{st_2} will be empty and the fixed points $(1,0)$ and $(0,1)$ will be always unstable. Thus the operating point will not converge to these undesirable fixed points in the steady-state. \\
Now consider the (possibly feasible) fixed point $(x_1,x_2)$ such that
\begin{eqnarray}
x_1=0, S_2=T
\end{eqnarray}
The above condition translates to,
\begin{eqnarray}
x_1=0, A_2 x_2 \beta = T \nonumber \\
x_1=0, x_2=\frac{T}{A_2\beta}
\end{eqnarray}
This fixed point will be feasible provided $\frac{T}{A_2\beta}<1$. The Jacobian matrix about this fixed point is evaluated as 
\begin{eqnarray*}
\frac{\partial \bm{F}(\bm{x})}{\partial{\bm{x}}}|_{\bm{x}}=-
\begin{pmatrix}
(1-\beta)\frac{T}{\beta}-T && 0 \\
\frac{T}{A_2\beta}(1-\frac{T}{A_2\beta}) B_{21} && \frac{T}{A_2\beta}(1-\frac{T}{A_2\beta})B_{22}
\end{pmatrix}
\end{eqnarray*}
From the Jacobian matrix above, we conclude that both of its eigenvalues are negative provided the following two conditions hold

\begin{eqnarray}
A_2> \frac{T}{\beta} \hspace*{2pt}, \hspace*{2pt} \beta<\frac{1}{2}.
\end{eqnarray}
\newpage 
 Doing similar analysis around the fixed point $(S_1=T, x_2=0)$, we conclude that the above fixed point will be stable for all arrival rates $A_1> \frac{T}{\alpha}$ and $\alpha <\frac{1}{2}$. \\
Finally, to obtain efficient operating region for the system (with no uncontrollable overload situation), we take union over stability region of all undesired fixed points and take the complement of it. Hence, if $\alpha > \frac{1}{2},\beta > \frac{1}{2}$ all the undesrible fixed points are unstable and hence the uncontrollable overload situation is avoided for \emph{all} DNS-request arrival rates $\bm{A}$. This proves part (2) of the theorem. \\
 On the other hand, if either $\alpha < \frac{1}{2}$ or $\beta< \frac{1}{2}$ holds, then a sufficient condition for the stability of the system is given by 
\begin{eqnarray*}
A_1<\frac{T}{1-\alpha}, A_2<\frac{T}{1-\beta}
\end{eqnarray*}
This proves part (3) of the theorem. 
\end{proof}

\end{document}